\documentclass[11pt, reqno]{amsart}
\usepackage{amsmath, amsthm, a4, latexsym, amssymb}

\def\@rmrk#1#2{\refstepcounter
    {#1}\@ifnextchar[{\@yrmrk{#1}{#2}}{\@xrmrk{#1}{#2}}}

\makeatletter\@addtoreset{equation}{section}\makeatother

 \newfont{\bfit}{cmbxti10 scaled 2000}
 \newfont{\biggi}{cmr12 scaled 2000}

 \newcommand{\R}{\mathbb{R}}
 
 \newcommand{\N}{\mathbb{N}}

 \newcommand{\prob}{\mathbb{P}}

 \newcommand{\me}{\mathbb{E}}
 
 \renewcommand{\P}{\mathbb{P}}
 
 \newcommand{\skrib}{{\mathcal B}}

 \newcommand{\skrih}{{\mathcal H}}

 \newcommand{\skril}{{\mathcal L}}
 \newcommand{\skrim}{{\mathcal M}}

 \newcommand{\skris}{{\mathcal S}}

 \newcommand{\skriy}{{\mathcal Y}}

 \newcommand{\sfrac}[2]{\mbox{$\frac{#1}{#2}$}}

\def\1{{\mathchoice {1\mskip-4mu\mathrm l}      
{1\mskip-4mu\mathrm l}
{1\mskip-4.5mu\mathrm l} {1\mskip-5mu\mathrm l}}}

\newcommand{\eq}{\begin{equation}}
\newcommand{\en}{\end{equation}}

\newenvironment{Proof}
{\vskip0.1cm\noindent{\bf Proof. }{\hspace*{0.3cm}}}%
{\nopagebreak {\hspace*{\fill}\rule{2mm}{2mm}}\\ }

{\nopagebreak {\hspace*{\fill}\rule{2mm}{2mm}}\\ }

\renewcommand{\subsection}{\secdef \subsct\sbsect}
\newcommand{\subsct}[2][default]{\refstepcounter{subsection}
\vspace{0.15cm}
{\flushleft\bf \arabic{section}.\arabic{subsection}~\bf #1  }
\nopagebreak\nopagebreak}
\newcommand{\sbsect}[1]{\vspace{0.1cm}\noindent
{\bf #1}\vspace{0.1cm}}

\newtheorem{theorem}{Theorem}[section]
\newtheorem{lemma}[theorem]{Lemma}

\newtheorem{prop}[theorem]{Proposition}

\newtheoremstyle{thm}{1.5ex}{1.5ex}{\itshape\rmfamily}{}
{\bfseries\rmfamily}{}{2ex}{}

\newtheoremstyle{rem}{1.3ex}{1.3ex}{\rmfamily}{}
{\itshape\rmfamily}{}{1.5ex}{}
\theoremstyle{rem}
\newtheorem{remark}{{\slshape\sffamily remark}}[]

\refstepcounter{subsection}

\def\thebibliography#1{\section*{References}
  \list%
  {\arabic{enumi}.}
    {\settowidth\labelwidth{[#1]}\leftmargin\labelwidth
    \advance\leftmargin\labelsep
    \parsep0pt\itemsep0pt
    \usecounter{enumi}}
    \def\newblock{\hskip .11em plus .33em minus .07em}
    \sloppy                   
    \sfcode`\.=1000\relax}

\begin{document}
\title[LLDP  and LDP  for  Signal -to- Interference  and Noise Ratio Graph  Models ]
{\Large  Large  Deviation Principle  for Empirical sinr  Measure of Critical  Telecommunication  Networks}

\author[]{}

\maketitle
\thispagestyle{empty}
\vspace{-0.5cm}

\centerline{\sc By  Enoch Sakyi-Yeboah,  Charles Kwofie  and   Kwabena Doku-Amponsah }
\renewcommand{\thefootnote}{}
\footnote{\textit {Affiliation:} University  of  Ghana, \textit {Email:}kdoku-amponsah@ug.edu.gh}
\renewcommand{\thefootnote}{1}

\vspace{-0.5cm}

\begin{quote}{\small }{\bf Abstract.}
For   a \emph{ powered Poisson  process},  we  define  \emph{Signal-to-Interference-plus-Noise Ratio}(SINR)  and  thesinr  network  as  a Telecommunication Network. We  define  the  Empirical  Measures   (\emph{empirical powered  measure}, \emph{empirical link measure}  and  \emph{empirical sinr measure})  of  a  class of  Telecommunication Networks. For  this  class  of  Telecommunication  Network  we  prove  a  joint  large  deviation principle for  the  empirical  measures  of  the  Telecommunication  Networks. All  our  rate  functions   are  expressed  in  terms  of  relative  entropies.
\end{quote}\vspace{0.3cm}

\vspace{0.3cm}
\vspace{0.3cm}


\section{Introduction   and Background}
\subsection{Introduction}

Since the inception of the 19th century the  world  had experience renaissance in the information theory of wireless communication channel, wireless networks. Further, multimedia technologies have seen significant growth in the last two decades. The use of handheld devices and obtaining services offered by the Internet has now become essential in our daily lives. Therefore, the availability of wireless networks and network quality of service (QoS) offer have become vital for mobile users.  See, Hassan, Tan and Yap~\cite{HTY2019}.\\

 Currently, telecommunication is simply an electrical medium of connecting over a distance (location and battery power). Telecommunication was discovered as an electrical waves and it  is suggested that they could travel a speed close to the speed of light. See, Paudel and  Bhattarai~\cite{PB2018}.The fundamental requirement of routing through any telecommunication network whether it is via voice call or data package; is that each end point on the network has a unique address which enables wireless communication. 
Cellular systems are now nearly universally deployed and are under ever-increasing pressure to increase the volume of data they can deliver to consumers. Refer  to  Andrews, Baccelli and  Ganti~\cite{ABG2011}  or  the  references  therein.

 Now, the advent of multimedia interactive services and the surge in the number of interconnected devices has led  to investigation of new approaches able to enhance wireless capacity in 5G networks. See,  example Aravanis, Lam, Muñoz, Pascual-Iserte  and Di Renzo~\cite{ALMPD2019}.  Marvi, Aijaz and  Khurram~\cite{MAK2019} posited that 8.3 billion hand-held devices and 3.3 billion machine-to-machine (M2M) devices will be connected by 2021. The number of connected devices would clearly exceed the expected global population of 7.8 billion by that time. The monthly global mobile data traffic is expected to reach 49 exabytes and the annual traffic will exceed half a zettabyte by 2021. \\

In informatory theory and telecommunication engineering, wireless consist of nodes which connect over a wireless channel,see   Gupta and  Kumar~\cite{GK2000}.  Signal -to- Interference-Plus-  Noise  Ratio (SINR) is a tool used as rate of information transfer in wireless communication system such as networks.  According to Jeske  and  Sampath ~\cite{JS2004}, thesinr is an important metric of wireless communication link quality.sinr estimates have several important applications. These include optimizing the transmit power level for a target quality of service, assisting with handoff decisions and dynamically adapting the data rate for wireless Internet applications. Communication performance can be improved significantly by adaptive transmissions based on the quality of received signals, i.e., the signal-to-interference-plus-noise ratio (SINR) ~Choi, Joo, Zhang  and  Shroff~ \cite{CJZS2013}. \\

 In cellular networks,sinr is a quantity that indicates if a given frequency resource is suitable to properly maintain a communication link. This is the rationale behind the usage insinr in network to monitor the occurrence of radio link and handover failures,  see Bastidas-Puga, Galaviz and  Andrade~\cite{BGA2018}. An accuratesinr estimation provides for both a more efficient system and a higher user‐perceived quality of service. Thus, thesinr is popularly used in wireless connection as a way to measure the quality of wireless connection within the space.\\

 Analogous to the SNR used often in wired communications systems, thesinr is defined as the power of a certain signal of interest divided by the sum of the interference power (from all the other interfering signals) and the power of some background noise. If the power of noise term is zero, then thesinr reduces to the signal-to-interference ratio (SIR). Conversely, zero interference reduces thesinr to the signal-to-noise ratio (SNR), which is used less often when developing mathematical models of wireless networks such as cellular networks.\\

Agrawal and Kshetrimavum~\cite{AK2017} derived the IPI, ISI and average outputsinr expressions for the binary phase shift keying (BSPK) modulated raised-cosine pulse in the beamforming-based mm-Wave MIMO system. At each receiving antenna, they use the coherent Rake receiver to capture the signal energy, carried by multipath components in the complete IEEE 802.15.3c channel model. Additionally, their paper presents the impacts of pulse duration and half power beam width (HPBWs) of the transmitting antennas on the average outputsinr.\\

Choi et al., (2013) focused on developing link scheduling schemes that can achieve optimal performance under thesinr model. The underlying argument was to treat an adaptive wireless link as multiple parallel virtual links with different signal quality, building on which they develop throughput-optimal scheduling schemes using a two-stage queueing structure in conjunction with recently developed carrier-sensing techniques. Furthermore, they introduce a novel three-way handshake to ensure, in a distributed manner, that all transmitting links satisfy theirsinr requirements.\\ 

Keeler, Błaszczyszyn  and  Karray~\cite{KBK2013} worked on an explicit integral interaction forsinr distribution experienced by a typical user in the downlink channel from the k-th strongest base stations of a cellular network modelled by Poisson point process on the plane. The outcome of their work shows that whole domain ofsinr was valid wheneverSinr$ <1,$ where one observes multiple coverage.\\

Aravanis et al.,~\cite{ALMPD2019} employs MGF to provide closed form expressions for the downlink ergodic capacity for the interference limited case, and validated the accuracy of these expressions by the use of extensive Monte Carlo simulations.\\

 Weiss~\cite{We1995} used large deviations techniques to analyze  models of communication networks. It  was  assumed  the  points  form  a sequence of independent and identical distributed random variables progressing  to some powerov processes in discrete or continuous time. Guiliano and Macci~\cite{GM2014} studied on the sequences of independent and identical distributed random variables and, under suitable conditions on the (common) distribution function, they proved large deviation principles for sequences of maxima, minima and pairs formed by maxima and minima. They assumed that  the independent and identical distributed random variables can be either unbounded or bounded; in the first case maxima and minima have to be suitably normalized.\\
 
 Duffy, Macci, and Torrisi~\cite{DMT2011} under suitable  assumptions about the large deviation behavior of the state selection and sojourn processes, proved that the empirical laws of the phase process satisfy a sample path large deviation principle. From this large deviation principle, the large deviations behavior of a class of modulated additive processes was deduced and an alternate proof of results for modulated Lévy processes  were obtained. With a practical application of the results,  they calculated the large deviation rate function for a process that arises as the International Telecommunications Union's standardized stochastic model of two-way conversational speech.  Other  researcher  have  also  studied  large  deviation  behaviour of the  interference  in  a wireless communication model. See, Ganesh  and  Torrisi~\cite{GT2008}.\\

 In this article, we prove a joint large deviation principle for the  empirical powered  measure,  empirical link  measure and the empiricalsinr measure   of critical  Telecommunication networks. In this sequel, we prove  a joint large  deviation  principle for  the   empirical powered  measure  and empiricallink  measure  of  the  Critical  Telecommunication Network.  See, Sekyi-Yeboah, Asiedu and  Doku-amponsah~\cite{SAD2020} for  similar  results for  the  dense Telecommunication Networks.   The  main  techniques  used  in this  article  are  the  Gartner-Ellis Theorem,  see Dembo and Zeitouni~\cite[Theorem~2.3.6]{DZ1998}   and the method  of  Mixtures as  deployed  in  the Ph  Thesi, Doku-Amponsah~\cite{DA2006}. \\
 
 The  remaining   part  of  this  article  is  organized in  the  following  manner.  Section~\ref{Sec2}  contain the  statement  of  the  main  results;  the  joint  LDP  for the  empirical  powered  measure, empirical link  measure  and  empiricalsinr  measure,  the  joint LDP  for  the  empirical  powered measure  and  the  empirical link  measure,  and  the  conditional  LDP  for  the  empiricalsinr measure  given  the  empirical powered  measure  aand  the  empirical link  measure. In  Section~\ref{Sec3}  we  give  the  proofs  of  Theorem~\ref{main1b}(i)~and Theorem~\ref{main1abc}. Section~\ref{Sec4} gives  the  proofs  of  Theorem~\ref{main1b}(ii)  and Theorem~\ref{main1ab}.


\subsection{Background}\label{SINR}\label{Sec1}

For  a  fix  dimension  $d\in\N$    and  a  measureable  set  $D\subset \R^d$    with  respect  to  the  Borel-Sgma  algebra  $\skrib(\R^d).$ 
Let   $\lambda \eta:D \to [0,1]$, be  a rate measure,  $q$   a transition kernel from  $D$   to  $(0,\,\infty)$  and      $\phi(r)=r^{-\ell}, $   here  $\ell \in(0,\infty),$   be  a  path loss  function and   $\tau^{(\lambda)}, \gamma^{(\lambda)}:(0\,,\,\infty)\to (0\,,\,\infty),$   be technical constants. We  define  theSinr  Graph as  follows:

\begin{itemize}
\item  Pick  $Y=(Y_i)_{i\in I}$   a  Poisson  Point  process  (PPP)  with rate measure  $\lambda \eta:D \to [0,1]$.
\item For  $Y,$   we    assign  each  $Y_i  $   a  power  $\rho(Y_i)=\rho_i$  independently  according  to  the  transition  function   $q(\cdot \,,\,Y_i).$
\item For   any  two  powered  points  $((Y_i,\rho_i),(Y_j,\rho_j))$  we    connect  a  link  iff  $$sinr(Y_i,Y_j, Y) \ge \tau^{(\lambda)}(\rho_j) \mbox{  and      $SINR(Y_j,Y_i, Y)\ge \tau^{(\lambda)}(\rho_i),$}$$  where $$SINR(Y_j,Y_i, Y)=\frac{\rho_i \phi(\|Y_i-Y_j\|)}{N_0 +\gamma^{(\lambda)}(\rho_j)\sum_{i\in I-\{j\} }\rho_i\phi(\|Y_i-Y_j\|)}$$
\end{itemize}

We shall  consider  $Y^{\lambda}(\eta, Q, \phi)=\Big\{[(Y_i,\rho_i), j\in I], \, E\Big\}$  under  the  joint  law  of  the  powered  Poisson Point  process  and    the  network.  We  will  interpret   $Y^{\lambda}$    as   an Sinr  network  and   $ (Y_i,\rho_i):= Y_i^{\lambda}$   as  the  power   of    device  $i.$  We  recall  Proposition~\ref{Pro}   from  \cite{SAD2020}  as  follows:

\begin{prop}[SAD2020]\label{Pro}
	
	The link  probability  of  the Sinr  network,   $ p_{\lambda}$,  is  given  by  $$ p_{\lambda}((x,\rho_x),(y,\rho_y))= e^{-\lambda h_{\lambda}^{D} ((x,\rho_x),(y,\rho_y))) },$$
	
	$$h_{\lambda}^{D} ((x,\rho_x),(y,\rho_y))= \int_{D} \Big[\sfrac{ \tau^{(\lambda)}(\rho_x) \gamma^{(\lambda)}(\rho_x)  }{\tau^{(\lambda)}(\rho_x) \gamma^{(\lambda)}(\rho_x)+(\|z\|^{\ell}/\|x-y \|^{\ell})} + \sfrac{ \tau^{(\lambda)}(\rho_y) \gamma^{(\lambda)}(\rho_y)  }{\tau^{(\lambda)}(\rho_y) \gamma^{(\lambda)}(\rho_y)+(\|z\|^{\ell}/\|y-x \|^{\ell})}\Big] \eta(dz).$$
\end{prop}

We  assume  there  a  sequence  of  real  numbers  $a_{\lambda}$  and  a function  $h_* :D\times \R_+\to (0,\infty)$ such  that  $\lambda^{2}a_{\lambda}\to\infty$  and  $$\lim_{\lambda\uparrow\infty}\lambda^2a_{\lambda}p_{\lambda}((x,\rho_x),(y,\rho_y))=h_{*}((x,\rho_x),(y,\rho_y)).$$

We  shall  call  $Y^{\lambda}$   \emph{Critical}Sinr  if   $\lambda a_{\lambda}\to 1$,
\emph{Subcritical}sinr  if  $\lambda a_{\lambda}\to  0$ and  if  $\lambda a_{\lambda}\to  \infty$  we  call it \emph{Supercritical}sinr.   In  this  article  We  shall  look  at  criticalsinr networks. i.e.$$\lim_{\lambda\to\infty}\lambda a_{\lambda}\to 1.$$  We  define  the  set $\skris(D)$  by

 \begin{equation}
\skris(D)=\cup_{y\subset D}\Big\{y:\,\, |y\cap B|<\infty\,\, ,\mbox{for\, any  bounded  $B\subset D$ }\Big \}.
\end{equation}

Write  $\skriy= \skris(D\times\R_+)$   and  denote  by  $\skrim(\skriy)$, the  space  of  positive  measures  on  the  space  $\skriy$   equipped  with  $\tau-$  topology. Henceforth,  we shall call  $\skriy$ a locally  finite  subset  of  the  set  $\skriy.$\\

 {\bf Empirical  measures of  the  Sinr  Networks:}   For any Sinr graph $Y^\lambda$  we  define a probability measure, the
\emph{empirical power measure}, ~ $M_1^{\lambda}\in\skrim(\skriy)$,~by
$$M_1^{\lambda}\big((x,\rho_x)\big ):=\frac{1}{\lambda}\sum_{i\in I}\delta_{Y_i^{\lambda}}\big((x,\rho_x)\big)$$
and a symmetric finite measure, the \emph{empirical pair measure}
$M_2^{\lambda}\in\skrim(\skriy\times \skriy),$ by
$$M_2^{\lambda}\big((x,\rho_x),(y,\rho_y)\big):=\frac{1}{\lambda}\sum_{(i,j)\in E}[\delta_{(Y_i^{\lambda},Y_j^{\lambda})}+
\delta_{(Y_j^{\lambda},Y_i^{\lambda})}]]\big((x,\rho_x),(y,\rho_y)\big).$$

 Note  that the  total mass  $\|M_1^{\lambda}\|$ of  the  empirical power  measure  is $\1$  and  total  mass  of  
the empirical link measure is
$2|E|/\lambda^2$.

\begin{theorem}[SAD2020]\label{main1a}
	Suppose   $Y^{\lambda}$  is  an Sinr  network  with rate measure
	$\lambda \eta:D \to [0,1]$ and   a   power  probability  function  $q$  from  $D$   to  $(0,\infty)$  and  path  loss  function   $\phi(r)=r^{-\ell}, $  for  $\ell>0 .$ 
	Then, as $\lambda\rightarrow\infty,$
	$M_1^{\lambda}$  satisfies an  LDP   in   the  space 
	$\skrim(\skriy)$
	with good rate function
	$$
	\begin{aligned}
	I_1(\sigma)= \left\{\begin{array}{ll}H(\sigma\,|\eta\otimes q),  & \mbox{if  $\|\sigma\|=1$ }\\
	\infty & \mbox{otherwise.}		
	\end{array}\right.
	\end{aligned}	$$
\end{theorem}

\section {Statement of Main  Results}\label{mainresults}\label{Sec2}
 We,  define the empirical   sinr measure  $\skril_{(x,\rho_x)}^{1,2} $ by $$\skril_{(x,\rho_x)}^{1,2}(a):=\frac{1}{\|N_{\lambda}((x,\rho_x))\|}\sum_{i\in N_{\lambda}([x,\rho])}\delta_{SINR(Y_i^{\lambda}, (x,\rho_x),M_1^{\lambda})}(a) ,$$
 
 where  $$N_{\lambda}(x)=\Big\{k\in I:\,  (k,j)\in E,\, Y_j^{\lambda}=x\Big\}$$

We observe  that  we  have    $$\|N_{\lambda}((x,\rho_x))\|=\lambda\int_{\skriy}M_2^{\lambda}((x,\rho_x),(dy,d\rho_y))$$

and  $$\skril_{(x,\rho_x)}^{1,2}(a)=\frac{1}{\| N^{\lambda}((x,\rho_x))\|}\int_{ \skriy} \Phi_{a}^{\lambda}\Big((x,\rho_x)\, ,\, (y,\rho_y)\,,\, M_1^{\lambda}\Big) M_2^{\lambda}((dy,d\rho_y),(dx,d\rho_x)),$$ where $$\Phi_{a}^{\lambda}\Big((x,\rho_x)\,,\, (y,\rho_y)\,,\,\sigma\Big)= \1_{\Big\{\tau^{(\lambda)}(\rho_{y})\le Sinr\big((y,\rho_y), (x,\rho_x)\,,\,\sigma\big)\le a(\rho_x)\Big\}} $$ 

We  write  $$\skril^{1,2}:=\Big(\skril_{(x,\rho_x)}^{1,2},(x,\rho_x)\in \skriy\Big).$$


Theorem~\ref{main1ab},  is  a  Joint  Large  deviation principle   for  the  empirical  measures  of  the Sinr  network models.  We  recall  from  Subsection~\ref{Sec1} the  definition  of  $h_{\lambda}^{D}$  as  

  $$h_{\lambda}^{D} ((x,\rho_x),(y,\rho_y))= \int_{D} \Big[\sfrac{ \tau^{(\lambda)}(\rho_x) \gamma^{(\lambda)}(\rho_x)  }{\tau(\rho_x) \gamma(\rho_x)+(\|z\|^{\ell}/\|x-y \|^{\ell})} + \sfrac{ \tau^{(\lambda)}(\rho_y) \gamma^{(\lambda)}(\rho_y)  }{\tau^{(\lambda)}(\rho_y) \gamma^{(\lambda)}(\rho_y)+(\|z\|^{\ell}/\|y-x \|^{\ell})}\Big] \eta(dz)$$    and   write	 $$h_*\sigma\otimes\sigma((x,\rho_x),(y,\rho_y))):=h_*((x,\rho_x),(y,\rho_y))\sigma((x,\rho_x))\sigma((y,\rho_y)).$$  
 We  define    $\Big\langle  f_{\sigma},\, \eta \Big\rangle$  by $$\Big\langle  f_{\sigma},\, \eta \Big\rangle_{(x,\rho_x)}(a) =\frac{1}{\eta_2((x,\rho_x))}\int_{\skriy} f_a((x,\rho_x),dz,\sigma)\eta(dz\,,\,(x,\rho)).$$

   Observe  that, for  a finite  measure $\omega$

\begin{equation}\label{Com2}
\Big\langle  \Phi_{\sigma}^{\lambda},\, \omega \Big\rangle_{(x\, ,\,\rho_x)}(a)=\frac{1}{\omega_2((x,\rho_x))}\int_{\skriy}\Phi_{a}^{\lambda}\Big((y,\rho_y)\,,\, (x,\rho_x)\,,\,\sigma\Big)\omega((dy,d\rho_y)\,,\,(x,\rho_x)), 	
\end{equation}

is   a  probability  measure.  We  write  $\displaystyle \lim_{\lambda\to\infty}\Phi_{a}^{\lambda}=\Phi_{a}$   and  note that   $$\displaystyle \lim_{\lambda\to\infty}\Big\langle  \Phi_{\sigma}^{\lambda},\, \omega \Big\rangle_{(x,\rho_x)}(a)=\Big\langle  \Phi_{\sigma},\, \omega \Big\rangle_{(x,\rho_x)}(a),\,\mbox{  for all $a\in[\tau,\infty),$}$$
by  he  dominated  convergence.

\begin{theorem}	\label{main1ab}
	Let   $Y^{\lambda}$  is  a critical powered Sinr  network  with rate measure
	$\lambda \eta:D \to [0,1]$ and   a   power  probability  function  $q$  from  $D$   to  $(0,\,\infty)$  and  path  loss  function   $\phi(r)=r^{-\ell}, $  for  $\ell>.$   Suppose  $q$ is an exponential  distribution  with  mean  $1/c.$ Then, as  $\lambda\to \infty$,  the  triplet  $(M_1^{\lambda},M_2^{\lambda}, \skril^{1,2})$  satisfies a  large  deviation principle  in the  space 
		$\skrim(\skriy)\times \skrim(\skriy\times\skriy)\times\skrim(  [\tau,\,\infty))$
		with  speed   $\lambda$  and  good rate function

		\begin{equation}
		J_*\big(\sigma, \omega,\nu\big)= H\Big(\sigma\Big |\eta\otimes q\Big)+\frac{1}{2} \skrih\Big(\omega\|h_*\sigma\otimes\sigma\Big)+\frac{1}{2}\int_{\skriy}H\Big(\nu_{(x,\rho_x)}\Big\|\Big\langle \Phi_{\sigma} ,\omega \Big\rangle_{[x\, ,\,\rho_x]}
		\Big) \omega_2((dx,d\rho_x))
		\end{equation} 
	\end{theorem}

\begin{remark}{\bf Interpretation  of  the  Rate  Function:}
	The  rate  function  can  be  regarded  as  the  cost  of  having  a  powered   Sinr network  corresponding  to  the  empirical  measures triplet  $(\sigma, \omega,\nu).$   This  cost  may  be  divided  into  three separate  costs:
	\begin{itemize}
		\item[(i)]The  first  term  is  the  cost  of  having  the  empirical powered  measure  $\sigma$  is  known.  This  cost  is  non-negative  and  it is  zero  iff  $\sigma=\eta\otimes q.$
		\item[(ii)] The  second  term  is  the  cost  of  obtaining  the  empirical link  measure  $\omega$  given the  empirical  powered  measure  $\sigma$.  This  cost  is  also  non-negative  and  it  is  zero  iff  $ \omega=h_{*}\sigma\otimes\sigma$
		\item[(iii)] The  last  term  is  the  cost  obtaining  the  empiricalsinr measure  given  the  empirical powered  measure  $\sigma$  and  the  empirical link  measure  $\omega.$  This  cost  is  also  none-negative  and  it  is  zero  iff   $\nu=\Big\langle \Phi_{\sigma} ,\omega \Big\rangle.$

	\end{itemize}

	This  implies   the  cost   $J_*\big(\sigma, \omega,\nu\big)=0$  iff   $\sigma=\eta\otimes q,$  $ \omega=h_{*}(\eta\otimes q)\otimes (\eta\otimes q)$   and  $$\nu=\Big\langle \Phi_{\eta\otimes q} ,h_{*}(\eta\otimes q)\otimes (\eta\otimes q)\Big\rangle$$
\end{remark}

We  write  $B_x(s):=\big\{y:\|y-x\|<s\big\},$    $\displaystyle \skrib_{(\rho_x,\rho_y)}^{t}:=B_x\Big(\Big[\frac{c\rho_x}{t(\rho_y)\gamma(\rho_y)}\Big]^{1/\ell}\Big[\int_{D}\|z-x\|^{-\ell}\eta(dz)\Big]^{-1/\ell}\Big)$  and  note that  the  typical  behaviour of  the   empirical Sinr  measure   is  as   	
$$\begin{aligned}
\nu_{(x,\rho_x)}(a)&=\Big\langle \Phi_{\eta\otimes q} ,h_{*}(\eta\otimes q)\otimes (\eta\otimes q)\Big\rangle_{[x, \rho_x]}(a)\\
&= \int_{\skriy}\Phi_{a}\Big((x,\rho_x), (y,\rho_y), \eta\otimes q\Big)\frac{ h_{*}((x,\rho_x),[y,\rho])   \eta(dy) q(d\rho_y)}{\int_{\skriy}h_{*}((x,\rho_x),[y,\rho])\eta(dy)q(d\rho_y)}\\
&= \int_{\R_+} \int_{D} \1_{\skrib_{(\rho_x,\rho_y)}^{\tau}\setminus \skrib_{(\rho_x,\rho_y)}^{a}}(x)\frac{e^{-c\rho_y} h_{*}((x,\rho_x),[y,\rho])   \eta(dy) d\rho_y}{\int_{\R_+}\int_{D}e^{-c\rho_y}h_{*}((x,\rho_x),[y,\rho])\eta(dy) d\rho_y}
\end{aligned}$$

where  $B\setminus A=B\cap A^c$ and  $\1_{\Gamma}(x)$ denote  the  indicator  function on  the  set $\Gamma.$  

 Theorem~\ref{main1b} below   is  a  conditional large  deviation  principle  for  the  empirical link  measure given  the  empirical power  measure,  and  joint  LDP   for  the  empirical  power measure  and empirical link  measure  of the Sinr  network model.    We  define  a  relative  entropy  $\skrih$  by 

\begin{equation}
\begin{aligned}
	\skrih(\omega\|h_*\sigma\otimes\sigma):= \left\{\begin{array}{ll} H(\omega\,\|\,h_{*}\sigma\otimes \sigma)+\Big(\|h_{x}\sigma\otimes\sigma\|-\|\omega\|\Big),  & \mbox{if  $\|\omega\|>0.$  }\\
		\infty & \mbox{otherwise.}		
	\end{array}\right.
\end{aligned}
\end{equation} 

\begin{theorem}	\label{main1b}
		Let   $Y^{\lambda}$  is  an Sinr  network  with rate measure
		$\lambda \eta:D \to [0,1]$ and   a   powered  probability  function  $q$  from  $D$   to  $(0,\infty)$  and  path  loss  function   $\phi(r)=r^{-\ell}, $  for  $\ell>0.$   Suppose  $q$  is an  exponential  distribution  with  parameter $c.$
		
\begin{itemize}

\item[(i)]Then, as $\lambda\rightarrow\infty,$ conditional  on the  event  $M_1^{\lambda}=\sigma,$  $M_2^{\lambda}$  satisfies a  large  deviation principle in the  space 
$\skrim(\skriy\times \skriy)$
with  speed   $\lambda$  and  good rate function

\begin{equation}
I_{\sigma}(\omega)=\frac{1}{2}\skrih(\omega\|h_*\sigma\otimes\sigma)
\end{equation} 

\item[(ii)]Then as  $\lambda\rightarrow\infty,$ the pair
$(M_1^{\lambda},\, M_2^{\lambda})$  satisfies  a  large  deviation principle  in  the  space
$\skrim(\skriy\times \skriy)$
with  speed   $\lambda,$  and  good rate function

\begin{equation}
I(\sigma,\,\omega)=  H\Big(\sigma\,\Big |\eta\otimes q\Big)+\frac{1}{2} \skrih(\omega\|h_*\sigma\otimes\sigma),
\end{equation}

\end{itemize}

where $$h_*\sigma\otimes\sigma((x,\rho_x),(y,\rho_y)))=h_*((x,\rho_x),(y,\rho_y))\sigma((x,\rho_x))\sigma((y,\rho_y)).$$
\end{theorem}

\begin{theorem}	\label{main1abc}
	Let   $Y^{\lambda}$  is  a critical powered Sinr network  with rate measure
	$\lambda \eta:D \to [0,1]$ and   a   powered  probability  function  $q$  from  $D$   to  $(0,\infty)$  and  path  loss  function   $\phi(r)=r^{-\ell}, $  for  $\ell>0.$  Suppose  $Y$   is  an  Sinr notwork   conditional  on  the event  $\displaystyle \Big\{(M_1^{\lambda}, M_2^{\lambda})=(\sigma,\, \omega)\Big\}.$ Then, as  $\lambda\to \infty$,  the empirical sinr  measure   $ \skril^{1,2}$  satisfies an LDP in the  space 
	$\skrim(  [\tau,\,\infty))$
	with  speed   $\lambda$  and  good rate function
  $$\tilde{J}(\nu)= \frac{1}{2}\int_{\skriy}H\Big(\nu_{(x,\rho_x)}\Big\|\, \big\langle \Phi_{\sigma} ,\omega\big\rangle_{(x,\rho_x)}\Big)\omega_2((dx,d\rho_x)),$$
  
  where  $\omega_2$ denote   second  marginal  of  the  finite  measure $\omega.$
\end{theorem}

\section{ Proof  of  Theorem~\ref{main1b}  and  Theorem~\ref{main1ab}}\label{proofmain1b}\label{Sec3}

\subsection{Proof  of   Theorem~\ref{main1b}(i)  by  Gartner-Ellis  Theorem }\label{proofmain}

Suppose   $B_1,...,B_n$   is  a  decomposition  of  the  space  $D\times \R_{+}.$   Observe  that,  for  every $(x,y)\in B_i\times B_j,\, i,j=1,2,3,...,n,$  $\lambda M_2^{\lambda}(x,y)$  given  $\lambda M_1^{\lambda}(x)=\lambda\sigma(x)$  is  binomial  with  parameters  $\lambda^2\sigma(x)\sigma(y)/2$  and  $p_{\lambda}(x,y).$ Let  $q$  be  the  exponential  distribution  with  parameter $c.$  We  recall  the  function  $R_{\lambda}^{D}  $  from  the  previous  sections  as  follows:

 $$h_{\lambda}^{D} ((x,\rho_x),(y,\rho_y))= \int_{D} \Big[\sfrac{ \tau^{(\lambda)}(\rho_x) \gamma^{(\lambda)}(\rho_x)  }{\tau^{(\lambda)}(\rho_x) \gamma^{(\lambda)}(\rho_x)+(\|z\|^{\ell}/\|x-y \|^{\ell})} + \sfrac{ \tau^{(\lambda)}(\rho_y) \gamma^{(\lambda)}(\rho_y)  }{\tau^{(\lambda)}(\rho_y) \gamma^{(\lambda)}(\rho_y)+(\|z\|^{\ell}/\|y-x \|^{\ell})}\Big] \eta(dz).$$

Lemma~\ref{main1c}  is  key  component    in  the  application  of  the  Gartner-Ellis  Theorem,  see  example,  
\begin{lemma}\label{randomg.LDM1b}\label{main1c}
	Let   $Y^{\lambda}$  be  an Sinr  network  with rate measure
	$\lambda \eta:D \to [0,1]$ and   a   powered  probability  function  $q$  from  $D$   to  $(0,\infty)$  and  path  loss  function   $\phi(r)=r^{-\ell}, $  for  $\ell>0 ,$  conditional  on the  event  $M_1^{\lambda}=\sigma.$  Let $ g:\skriy\times \skriy\to \R$ be  bounded  function.  Then,
	
	$$\begin{aligned}\lim_{\lambda\to\infty}\frac{1}{\lambda}\log\me \Big\{e^{\lambda\langle g, \, M_2^{\lambda}\rangle }\Big | M_1^{\lambda}=\sigma\Big\}&=\frac{1}{2}\lim_{n\to\infty}\sum_{j=1}^{n}\sum_{i=1}^{n}\Big\langle 1-e^{g},\, h_*\sigma\otimes\sigma\Big\rangle _{B_i \times B_j}\\
	&=\frac{1}{2}\Big\langle 1-e^{g},\, h_*\sigma\otimes\sigma\Big\rangle _{\skriy \times \skriy}.
	\end{aligned}$$
\end{lemma}	
\begin{Proof}
 Now  we  observe  that 
	$$ \me\Big \{ e^{\int \int \lambda g(x,y)M_2^{\lambda}(dx,dy)/2} \Big |M_1^{\lambda}=\sigma\Big \}= \me\Big\{\prod_{x \in \skriy} \prod_{y \in \skriy} e^{\lambda g(x,y)M_2^{\lambda}(dx,dy)/2} \Big\}$$
	
	$$ \me\Big \{\prod_{x \in \skriy} \prod_{y \in \skriy} e^{g(x,y)\lambda M_2^{\lambda}(dx,dy/2)} = \prod_{i=1} \prod_{j=1} \prod_{x \in B_i} \prod_{y \in B_j} \me\Big\{e^{g(x,y) \lambda M_2^{\lambda}(dx,dy)/2  }\Big \} $$
	
	$$ \log \Big \{e^{\lambda \langle g,M_2^{\lambda}\rangle/2}\Big|M_1^{\lambda}=\sigma \Big\} = \sum_{j=1}^{n} \sum_{i=1}^{n}\int_{B_j}\int_{B_i}\log\Big[1-p(x,y)+p(x,y)e^{g(x,y)}\Big]^{\lambda^{2} \sigma \otimes \sigma (dx,dy)/2} $$

	By the  dominated convergence theorem 
	$$   \frac{1}{\lambda} \log E \{e^{\lambda \langle g,M_2^{\lambda}\rangle/2 } \mid M_1^{\lambda}=\sigma \} = \frac{1}{\lambda}\sum_{j=1} \sum_{i=1} \int_{B_i} \int_{B_j} \log\Big [1-\big(1-e^{g(x,y)}) p_{\lambda}(x,y) + o (\lambda ) \Big]^{ \lambda^2 \sigma \otimes \sigma (dx,dy)/2} $$
	$$  \frac{1}{\lambda} \log \me \{e^{\lambda \langle g,M_2^{\lambda}\rangle /2} \| M_1^{\lambda}=\sigma \} = \lim_{\lambda \rightarrow \infty } \sum_{j=1} \sum_{i=1}\int_{B_i} \int_{B_j} \log \Big[1-(1-e^{g(x,y)}) p_{\lambda}(x,y) + o (\lambda ) \Big]^{ \lambda \sigma \otimes \sigma (dx,dy)/2} $$
	
	$$\lim_{\lambda \rightarrow \infty} \frac{1}{\lambda} \log \me \Big\{e^{\lambda \langle g,M_2^{\lambda}\rangle /2} \| M_1^{\lambda}=\sigma\Big \} = \frac{1}{2} \sum_{j=1} \sum_{i=1} \int_{B_i} \int_{B_j}\Big [(1-e^{g(x,y)})h_{*}(x,y)\sigma \otimes \sigma (dx,dy)\Big] $$
	
	$$\lim_{\lambda \rightarrow \infty}\frac{1}{\lambda} \log \me\{e^{\lambda \langle g,M_2^{\lambda}\rangle/2} \mid M_1^{\lambda}=\sigma \} 
	= \frac{1}{2}\sum_{j=1}^{n}\sum_{i=1}^{n} \Big\langle 1-e^{g},\, h_*\sigma\otimes\sigma\Big\rangle _{B_i \times B_j} $$
	
	$$ \begin{aligned}
	\lim_{n \rightarrow \infty} \lim_{\lambda \rightarrow \infty} \frac{1}{\lambda} \log \me \{e^{\lambda \langle g,M_2^{\lambda}\rangle/2 } \Big |M_1^{\lambda}=\sigma \}& =\frac{1}{2} \lim_{n \rightarrow \infty } \sum_{j=1}^{n}\sum_{i=1}^{n} \Big\langle 1-e^{g},\, h_*\sigma\otimes\sigma\Big\rangle _{B_i \times B_j}\\
	&	=\frac{1}{2} \Big\langle 1-e^{g},\, h_*\sigma\otimes\sigma\Big\rangle _{\skriy \times \skriy}  
\end{aligned}$$

Hence,	by Gartner-Ellis theorem, conditional  on the  event $\Big\{M_{1}^{\lambda}= \sigma\Big\}$, $M_2^{\lambda}$ obey a  large  deviation  principle with speed $\lambda$  and  rate function
		$$ I_{\sigma}(\omega) = \frac{1}{2}\sup_{g} \Big\{  \Big\langle g,\, \omega\Big\rangle _{\skriy \times \skriy}+  \Big\langle 1-e^{g},\, h_*\sigma\otimes\sigma\Big\rangle _{\skriy\times \skriy}\Big\}$$ 
	
	 which   when solved,  see  example   \cite{DA2006},	would clearly   reduces  to  the good rate function  given by 
		\begin{equation}
		I_{\sigma}(\omega)=  \frac{1}{2}\skrih(\omega\|h_* \sigma\otimes\sigma).
	\end{equation}
		
\end{Proof}

\subsection{Proof  of  Theorem~\ref{main1abc}  by  Gartner-Ellis  Theorem}

The  first  step  in  proof  of  Theorem~\ref{main1ab} is a  large  deviation principle  for  the  sequence  of  measures $\Big(\skril_{(x,\rho_x)}^{1,2},(x,\rho_x)\in\skriy\Big)$  conditional  on  the set   $$\Big\{(M_1^{\lambda},M_2^{\lambda})=(\sigma,\,\omega)\Big\}.$$

\begin{lemma}\label{Com1a}
	Let   $Y^{\lambda}$  be  a critical powered Sinr  networks  with rate measure
	$\lambda \eta:D \to [0,1]$ and   a   powered  probability  function  $q$  from  $D$   to  $(0,\infty)$  and  path  loss  function   $\phi(r)=r^{-\ell}, $  for  $\ell>0.$   Then,  for  every $(x,\rho_x)$,
	we  have  
	\begin{equation}\label{Com3a}
	\lim_{\lambda\to \infty}\P\Big\{\tau^{(\lambda)}(\rho_x)\le Sinr([Y_i,\rho_i], (x,\rho_x),\, M_1^{\lambda})\le a(\rho_x), \, (x,\rho_x)\in \skriy\Big |(M_1^{\lambda},\,M_2^{\lambda})=(\sigma,\omega)\Big\}=\Big\langle \Phi_{\sigma}\, ,\,\omega \Big\rangle (a)	\end{equation}
\end{lemma}
\begin{proof}
	
	We  compute the  probability 
	$$\begin{aligned}
	\P\Big\{\tau^{(\lambda)}(\rho_x)&\le Sinr([Y_i,\rho_i], (x,\rho_x),\, M_1^{\lambda})\le a(\rho_x)\Big |(M_1^{\lambda},\,M_2^{\lambda})=(\sigma,\omega)\Big\}\\
&=\frac{1}{\omega_2((x,\rho_x))}\int_{\skriy}\Phi_{a}^{\lambda}\Big((x,\rho_x)\,,\, (y,\rho_y)\,,\,\sigma\Big)\omega([dy,\,d\rho_y], (x,\rho_x))\\
&=\Big\langle \Phi_{\sigma}\, ,\,\omega \Big\rangle_{(x\, ,\,\rho_x)} (a)	\end{aligned}$$
\end{proof}
Taking  limits  as  $\lambda\to\infty$   on  both sides    we  have   \ref{Com3a}  which  ends  the  proof  of  Lemma~\ref{Com1a}.
\begin{lemma}\label{Com1}
		Let   $Y^{\lambda}$  is  a critical powered Sinr  graph  with rate measure
	$\lambda \eta:D \to [0,1]$ and   a   powered  probability  function  $q$  from  $D$   to  $(0,\infty)$  and  path  loss  function   $\phi(r)=r^{-\ell}, $  for  $\ell>0.$   Suppose  $q$  is  an  exponential  distribution  with  parameter $c.$ Then,  for  every $(x,\rho_x)$,
	we  have  	
\begin{equation}\label{Com3}
	\lim_{\lambda\to\infty}\frac{1}{\lambda}\log\me \Big\{e^{\int_{\skriy}N_{\lambda}((dx,d\rho_x))\langle g, \, \skril_{(x,\rho_x)}^{1,2}\rangle }\Big |( M_1^{\lambda},\, M_2^{\lambda})
=(\sigma,\,\omega)\Big\}=\frac{1}{2}\Big\langle \log\Big\langle e^{g}\,,\,\Big\langle \Phi_{\sigma}\, ,\,\omega \Big\rangle_{\cdot} \Big\rangle_{[\tau,\infty)} \, ,\, \omega_2\Big\rangle_{\skriy}
\end{equation}
\end{lemma}

\begin{proof}
	We  observe  that   $\Big(\skril_{(x,\rho_x)}^{1,2}|N((x,\rho_x)),\,(x,\rho_x)\in\skriy\Big)$ are  independent    distributed  as   $$\Big(\Big\langle \Phi_{\sigma}^{\lambda}\, ,\,\omega \Big\rangle_{(x\, ,\,\rho_x)},\,(x,\rho_x)\in\skriy\Big).$$  
	\begin{equation}\label{Com4}
	\begin{aligned}\me \Big\{e^{\int_{\skriy}\Big\langle g, \, \skril_{(x,\rho_x)}^{1,2}\Big\rangle N_{\lambda}((dx,d\rho_x))/2 }\Big |( M_1^{\lambda},\, M_2^{\lambda})&=(\sigma,\,\omega)\Big\}\\
	&=\prod_{(dx,d\rho_x)\in\skriy}\me_{\Big\langle \Phi_{\sigma}^{\lambda}\, ,\,\omega \Big\rangle_{[x\, ,\,\rho_x]}}\Big[\prod_{i\in N_{\lambda}([dx,dx])/2}e^{g(SINR(Y_i^{\lambda},(x,\rho_x),\sigma))}\Big]\\
	&=\prod_{(dx,d\rho_x)\in\skriy}\Big(\me_{\Big\langle \Phi_{\sigma}^{\lambda}\, ,\,\omega \Big\rangle_{[x\, ,\,\rho_x]}}\Big[e^{g(SINR(Y_i^{\lambda},(x,\rho_x),\sigma))}\Big]\Big)^{N_{\lambda}((dx,d\rho_x))/2}\\
	&=\prod_{(dx,d\rho_x)\in\skriy}\Big(\int_{\tau}^{\infty} e^{g(a)}\, \Big\langle \Phi_{\sigma}^{\lambda}\, ,\,\omega \Big\rangle_{(x\, ,\,\rho_x)} (da)\Big)^{N_{\lambda}((dx,d\rho_x))/2}
	\end{aligned}
	\end{equation}
	
	Now  taking limit of   normalized logarithm  of \ref{Com3}   and  observing  that  $N_{\lambda}((dx,d\rho_x))/\lambda\to \omega_2((dx,d\rho_x)),$    $\Big\langle \Phi_{\sigma}^{\lambda}\, ,\,\omega \Big\rangle_{(x\, ,\,\rho_x)}\to\Big\langle \Phi_{\sigma}\, ,\,\omega \Big\rangle_{(x\, ,\,\rho_x)}$ as $\lambda\to\infty$   we  have \ref{Com3},  which  ends  the  proof  of  Lemma~\ref{Com1}
\end{proof}

Now,  by  the  Gartner-Ellis Theorem, Conditional  on the  event  $\Big\{(M_1^{\lambda},M_1^{\lambda})=(\sigma,\omega)\Big\}$,the  probability measure  $\skril^{1,2} $   obeys  an  LDP  with  speed  $\lambda$   and  rate  function  $$\tilde{J}(\nu)=\frac{1}{2}\sup_{g}\Big\{\Big\langle \Big\langle g,\,\nu_{\cdot}\Big\rangle_{[\tau,\infty]},\,\omega_2\Big\rangle_{\skriy} -\Big\langle \log\Big\langle e^{g}\,,\,\Big\langle \Phi_{\sigma}\, ,\,\omega \Big\rangle_{\cdot} \Big\rangle_{[\tau,\infty)} \, ,\, \omega_2\Big\rangle_{\skriy}\Big\}.$$

Using  the  variational  formulation  of  relative  entropy  we  have  that  
$$\tilde{J}(\nu)=\frac{1}{2}\int_{\skriy} H\Big(\nu_{(x,\rho_x)}\Big\|\,\Big\langle \Phi_{\sigma}\, ,\,\omega \Big\rangle_{(x\, ,\,\rho_x)}\Big)\omega_2((dx\, ,\,d\rho_x)),$$    which  proves  Theorem~\ref{main1ab}.

\section{ Proof of  Theorem~\ref{main1a}(ii)   and  Theorem~\ref{main1ab}  by  Method  of  Mixtures}\label{Sec4}
For any $\lambda\in (0,\infty)$ we define
$$\begin{aligned}
\skrim_{\lambda}(\skriy) & := \Big\{ \sigma\in \skrim(\skriy) \, : \, \lambda\sigma(x) \in \N \mbox{ for all } x\in \skriy\Big\},\\
\tilde \skrim_{\lambda }(\skriy\times \skriy) & := \Big\{ \omega\in
\tilde\skrim_*(\skriy\times \skriy) \, : \, 
\lambda \,\omega(x,y) \in \N,\,  \mbox{ for all } \, x,y\in \skriy
\Big\}\, .
\end{aligned}$$

We denote by
$\Theta_{\lambda}:=\skrim_{\lambda }(\skriy)$
and
$\Theta:=\skrim(\skriy)$.
With
$$\begin{aligned}
P_{ \sigma_{\lambda}}^{(\lambda)}(\eta_{\lambda}) & := \prob\big\{M_2^{\lambda}=\eta_{\lambda} \, \big| \, M_1^{\lambda}=\sigma_{\lambda}\big\}\, ,\\
P^{(\lambda)}(\sigma_{\lambda}) & :=
\prob\big\{M_1^{\lambda}=\sigma_{\lambda}\big\}
\end{aligned}$$
$$P_{(\sigma_{\lambda},\omega_{\lambda})}^{(\lambda)}(\nu_{(x,\rho_x)}):=\P\Big\{\skril_{(x,\rho_x)}^{1.2}=\nu_{(x,\rho_x)} \big|(M_1^{\lambda},M_2^{\lambda})=(\sigma_{\lambda},\omega_{\lambda})\Big\}$$
the joint distribution of $M_1^{\lambda}$ and $M_2^{\lambda}$ is
the mixture of $P_{ \sigma_{\lambda}}^{(\lambda)}$ with
$P^{(\lambda)}(\sigma_{\lambda}),$   and  the  joint  distribution  of  $\skril^{1,2}$, $M_1^{\lambda}$  and  $M_1^{\lambda} $ is  a mixture  of  $\tilde{P}^{\lambda}$  with    $P_{(\sigma_{\lambda},\omega_{\lambda})}^{(\lambda)}$ as follows: 
\begin{equation}\label{randomg.mixture}
d\tilde{P}^{\lambda}( \sigma_{\lambda}, \eta_{\lambda}):= dP_{	\sigma_n}^{(\lambda)}(\eta_{\lambda})\, dP^{(\lambda)}( \sigma_{\lambda}).\,
\end{equation}

$$dP_{\lambda}(\nu, \sigma_{\lambda}, \eta_{\lambda}):=dP_{(\sigma_{\lambda},\omega_{\lambda})}^{(\lambda)}(\nu)d\tilde{P}^{\lambda}( \sigma_{\lambda}, \eta_{\lambda}).$$

(Biggins, Theorem 5(b), 2004) gives criteria for the validity of
large deviation principles for the mixtures and for the goodness of
the rate function if individual large deviation principles are
known. The following three lemmas ensure validity of these
conditions.

Observe   that the  family of
measures $({P}^{(\lambda)} \colon \lambda\in(0,\infty))$  is  exponentially tight on
$\Theta.$

\begin{lemma}[] \label{Com4}

	\begin{itemize}
	
\item[(i)] 	The  family of
	measures $(\tilde{P}^{\lambda} \colon \lambda\in(0,\infty))$  is  exponentially tight on
	$\Theta\times\tilde\skrim_*(\skriy\times \skriy).$

\item[(ii)] The  family	measures $(P_{\lambda} \colon \lambda\in(0,\infty))$  is  exponentially tight on
	$\Theta\times\tilde\skrim_*(\skriy\times \skriy)\times \skrim([\tau,\,\infty).$
\end{itemize}
\end{lemma}

Define the function
$I\colon{\Theta}\times\skrim_*(\skriy\times \skriy)\rightarrow[0,\infty],$  by

\begin{equation}
I(\sigma,\,\omega)=  H\Big(\sigma\,\Big |\eta\otimes q\Big)+ \skrih\Big(\omega\|h_*\sigma\otimes\sigma\Big)
\end{equation}

  and  recall  from  Theorem~\ref{main1abc}  that  
   $$\tilde{J}(\nu)= \frac{1}{2}\int_{\skriy}H\Big(\nu_{(x,\rho_x)}\,\Big\|\, \Big\langle \Phi_{\sigma}\, ,\,\omega \Big\rangle_{[x\, ,\,\rho_x]} 
  \Big)\omega_2((dx,d\rho_x)).$$

\begin{lemma}[]\label{Com5}
\begin{itemize}
	
\item[(i)]	$I$ is lower semi-continuous.
\item  [(ii)] $\tilde{J}$  is  lower  semi-continuous.
\end{itemize}
\end{lemma}

By (Biggins, Theorem~5(b), 2004) the two previous lemmas and the
large deviation principles we have established
Theorem~\ref{main1b}  and  Theorem~\ref{main1abc} ensure
that under $(\tilde{P}^{\lambda})$    and   $P_{\lambda}$ the random variables $(\sigma_{\lambda}, \eta_{\lambda})$   and $(\nu, \sigma_{\lambda}, \eta_{\lambda})$  satisfy a large deviation principle on
$\skrim(\skriy) \times \tilde\skrim(\skriy\times \skriy)$ and   	$\Theta\times\tilde\skrim_*(\skriy\times \skriy)\times \skrim([\tau,\,\infty) $  with good rate function  $I$   and $\tilde{J}$  respectively,  which  ends  the  proof of  Theorem~\ref{main1b}.


\end{document}